\pdfoutput=1
\documentclass{amsproc}
\usepackage{a4wide}
\usepackage{bookmark}

\usepackage{}

\newtheorem{theorem}{Theorem}[section]
\newtheorem{lemma}[theorem]{Lemma}

\theoremstyle{definition}

\theoremstyle{remark}
\newtheorem{remark}[theorem]{Remark}

\numberwithin{equation}{section}

\begin{document}

 \title[Asymptotic expansions for factorial moments]{Asymptotic expansions for factorial moments of some distributions \\  in the analysis of algorithms}


\author{Sumit Kumar Jha}
\address{Center for Security, Theory, and Algorithmic Research\\ 
International Institute of Information Technology, 
Hyderabad, India}
\curraddr{}
\email{kumarjha.sumit@research.iiit.ac.in}
\thanks{}


\date{}

\begin{abstract}
We establish asymptotic expansions for factorial moments of following distributions: number of cycles in a random permutation, number of inversions in a random permutation, and number of comparisons used by the randomized quick sort algorithm. To achieve this we use singularity analysis of certain type of generating functions due to Flajolet and Odlyzko. 
\end{abstract}

\maketitle

\section{Introduction}
We consider three different distributions arising in the analysis of algorithms in each of the following sections:
\begin{itemize}
\item[1.] Number of cycles in a random permutation.
\item[2.] Number of inversions in a random permutation
\item[3.] Number of comparisons used by the randomized quick sort algorithm.
\end{itemize}
For each of the above distributions, we establish asymptotic expansions for factorial moments using a unified framework developed by Kirschenhofer, Prodinger and Tichy in \cite{prodinger52}, and singularity analysis of certain type of generating functions due to Flajolet and Odlyzko in \cite{flajolet}.\par 
\textbf{Notation:} Let $S_{n}$ the set of all $n!$ permutations of $\{1,2,\cdots,n\}$. The Euler's constant is $\gamma=0.577\cdots $. Let $\mathfrak{R}_{p,q}(u)$ be an unspecified linear combination of terms of the form $$\left(\log\frac{1}{1-u}\right)^{i}(1-u)^{-j-1}$$ where $i,j$ are integers with either $j<q$ and $i$ is arbitrary, or $j=q$ and $i\leq p$ \par 
We would use the following result throughout:
\begin{theorem}[Flajolet and Odlyzko \cite{flajolet}]
Let 
\begin{equation}
f_{\alpha,\beta}(u)\equiv f(u)=\frac{1}{(1-u)^{\alpha}}\left(\log\frac{1}{1-u}\right)^{\beta},
\end{equation}
where $\alpha$ is a positive integer and $\beta$ is a non-negative integer. The coefficient of $u^{n}$ in $f(u)$, denoted $[u^{n}]f(u)$, admits the asymptotic expansion
\begin{equation}
\label{main}
[u^{n}]f(u)\sim \frac{n^{\alpha-1}}{(\alpha-1)!}(\log n)^{\beta}\left[1+\frac{C_{1}}{1!}\,\frac{\beta}{\log n}+\frac{C_{2}}{2!} \, \frac{(\beta)(\beta-1)}{\log^{2}(n)}+\cdots \right],
\end{equation}  
where 
$$C_{k}=(\alpha-1)!\left[\frac{d^{k}}{dx^{k}}\frac{1}{\Gamma(x)}\right]_{x=\alpha}$$
and $\Gamma(\cdot)$ is the gamma function.
\end{theorem}
\section{Number of cycles in a random permutation}
Let $C_{n,k}$ be the number of permutations in $S_{n}$ with exactly $k$ cycles. We use the following lemma from \cite{Stanley}:
\begin{lemma}
We have
\begin{equation}
\label{eq1}
C_{n,k}=(n-1)\cdot C_{n-1,k}+C_{n-1,k-1}\qquad (n,k\geq 1),
\end{equation}
with the initial conditions $C_{n,k}=0$ if $n<k$ or $k=0$, except $C_{0,0}=1$.
\end{lemma}
Let $s$th factorial moment of number of cycles in a random permutation be
$$\beta_{s}(n)=\sum_{k\geq 0}(k)_{s}\frac{C_{n,k}}{n!}$$
where $(k)_{s}=(k)(k-1)\cdots (k-s+1)$ and $(k)_{0}=1$.
\begin{theorem}
We have
$$\beta_{s}(n)=\log^{s}(n)+\gamma s\log^{s-1}(n)+O(\log^{s-2}(n))$$
for integers $s\geq 1$.
\end{theorem}
\begin{proof}
We start by defining the corresponding \emph{probability generating function}:
$$G_{n}(z)=\sum_{k\geq 0}\frac{C_{n,k}}{n!}z^{k},$$
with $G_{0}(z)=1$. After equation \eqref{eq1}, we have
\begin{eqnarray*}
G_{n}(z)=\sum_{k\geq 1}\frac{C_{n,k}}{n!}z^{k}=\frac{n-1}{n}\cdot \sum_{k\geq 1} \frac{C_{n-1,k}\cdot z^{k}}{(n-1)!}+\frac{z}{n}\cdot\sum_{k\geq 1} C_{n-1,k-1}\cdot \frac{z^{k-1}}{(n-1)!}\\
=\frac{z+n-1}{n}\cdot G_{n-1}(z)
\end{eqnarray*}
for $n\geq 1$. Next we consider the \emph{bi-variate generating function} defined by:
$$H(z,u)=\sum_{n\geq 0}G_{n}(z)u^{n}.$$ We then have
\begin{eqnarray*}
\frac{\partial }{\partial u}H(z,u)=\sum_{n\geq 1}n\cdot G_{n}(z)u^{n-1}=\sum_{n\geq 1}(z+n-1)\cdot G_{n-1}(z)u^{n-1}\\
=z\cdot H(z,u)+u\cdot \sum_{n\geq 1}(n-1)\cdot G_{n-1}(z)u^{n-2}\\
=z\cdot H(z,u)+u\cdot \frac{\partial H(z,u)}{\partial u}
\end{eqnarray*}
which gives 
\begin{equation}
\label{eq2}
\frac{\partial H(z,u)}{\partial u}=\frac{z}{1-u}\cdot H(z,u)
\end{equation}
with $H(1,u)=\sum_{n\geq 0}u^{n}=(1-u)^{-1}$.\par 
Solving equation \ref{eq2} gives
\begin{equation}
H(z,u)=(1-u)^{-z}.
\end{equation}
Now note that 
$$\beta_{s}(n)=\left[\frac{d^{s}}{dz^{s}}G_{n}(z)\right]_{z=1}.$$
The generating functions $f_{s}(u)$ of $\beta_{s}(n)$ are defined by
$$f_{s}(u)=\sum_{n\geq 0}\beta_{s}(n)u^{n}.$$
By the Taylor series we have
$$H(z,u)=\sum_{n\geq 0}\frac{G_{n}(z)}u^{n}=\sum_{n\geq 0}\sum_{s\geq 0}\frac{(z-1)^{s}}{s!}\beta_{s}(n)\cdot u^{n}=\sum_{s\geq 0}\frac{(z-1)^{s}}{s!}f_{s}(u)$$
which gives
$$f_{s}(u)=\left[\frac{\partial^{s}}{\partial z^{s}}H(z,u)\right]_{z=1}=\frac{1}{1-u}\left(\log\frac{1}{1-u}\right)^{s}.$$
Now using equation \eqref{main} (with $\alpha=1$ and $\beta=s$) we conclude the assertion after noting $\gamma=C_{1}.$
\end{proof}
\section{Number of inversions in a random permutation}
In a permutation $\sigma = (a_{1},a_{2},\cdots,a_{n})$ a pair $(a_{i}, a_j )$, $i < j$ is called inversion if $a_i > a_j$. Let $I_{n,k}$ be the number of permutations in $S_{n}$ having a total of $k$ inversions. Then we have from \cite{KnuthSort}
$$I_{n,k}=I_{n-1,k}+I_{n-1,k-1}+\cdots+I_{n-1,k-(n-1)}\quad (n,k\geq 1),$$
with the initial conditions $I_{n,k}=0$ if $n<k$ or $k=0$, except $I_{0,0}=1$.\par 
The corresponding \emph{probability generating function} is
$$G_{n}(z)=\sum_{k\geq 0}\frac{I_{n,k}}{n!}z^{k},$$
with $G_{0}(z)=1$. Then we have 
\begin{eqnarray*}
G_{n}(z)=\sum_{k\geq 0}\frac{I_{n,k}}{n!}z^{k}=\sum_{k\geq 0}(I_{n-1,k}+I_{n-1,k-1}+\cdots+I_{n-1,k-(n-1)})\cdot \frac{z^{k}}{n!}\\
=\frac{1+z+z^{2}+\cdots +z^{n-1}}{n}\cdot G_{n-1}(z)\\
=\frac{1-z^{n}}{n(1-z)}G_{n-1}(z).
\end{eqnarray*}
for $n\geq 1$. Next we consider the \emph{bi-variate generating function} defined by:
$$H(z,u)=\sum_{n\geq 0}G_{n}(z)u^{n}.$$
We then have
\begin{eqnarray*}
\frac{\partial}{\partial u}H(z,u)=\sum_{n\geq 0}nG_{n}(z)u^{n-1}=\sum_{n\geq 0}\frac{1-z^{n}}{1-z}G_{n-1}(z)u^{n-1}=\frac{H(z,u)}{1-z}-\frac{z\cdot H(z,zu)}{1-z}
\end{eqnarray*}
which gives 
\begin{equation}
\label{eq3}
(1-z)\cdot \frac{\partial}{\partial u}H(z,u)=H(z,u)-z\cdot H(z,zu)
\end{equation}
with $H(1,z)=\sum_{n\geq 0}u^{n}=(1-u)^{-1}$.\par 
Let $s$th factorial moment of number of inversions in a random permutation be
$$\beta_{s}(n)=\sum_{k\geq 0}(k)_{s}\frac{I_{n,k}}{n!}$$
where $(k)_{s}=(k)(k-1)\cdots (k-s+1)$ and $(k)_{0}=1$. We would prove here that
\begin{theorem}
We have
$$\beta_{s}(n)=\frac{n^{2s}}{4^{s}}+\frac{s(2s-11)}{9\cdot 4^{s}}n^{2s-1}+O(n^{2s-2})$$
for integers $s\geq 1$.
\begin{proof}
First note that
$$\beta_{s}(n)=\left[\frac{d^{s}}{dz^{s}}G_{n}(z)\right]_{z=1}.$$
The generating functions $f_{s}(u)$ of $\beta_{s}(n)$ are defined by
$$f_{s}(u)=\sum_{n\geq 0}\beta_{s}(n)u^{n}.$$
By the Taylor series we have
$$H(z,u)=\sum_{n\geq 0}G_{n}(z)u^{n}=\sum_{n\geq 0}\sum_{s\geq 0}\frac{(z-1)^{s}}{s!}\beta_{s}(n)\cdot u^{n}=\sum_{s\geq 0}\frac{(z-1)^{s}}{s!}f_{s}(u).$$
After equation \eqref{eq3} we have
\begin{eqnarray*}
\sum_{s\geq 0}\frac{(z-1)^{s+1}}{s!}f'_{s}(u)=z\cdot H(z,zu)-H(z,u) \\
=z\cdot \sum_{s\geq 0}\frac{(z-1)^{s}}{s!}f_{s}(uz)-\sum_{r\geq 0}\frac{(z-1)^{r}}{r!}f_{r}(u)\\
=z\cdot \sum_{s\geq 0}\frac{(z-1)^{s}}{s!}\sum_{m\geq 0}\frac{f_{s}^{(m)}(u)(z-1)^{m}u^{m}}{m!}-\sum_{r\geq 0}\frac{(z-1)^{r}}{r!}f_{r}(u)\\
=\sum_{i\geq 0}a_{i}(z-1)^{i}\cdot \sum_{s\geq 0}\frac{(z-1)^{s}}{s!}\sum_{m\geq 0}\frac{f_{s}^{(m)}(u)(z-1)^{m}u^{m}}{m!}-\sum_{r\geq 0}\frac{(z-1)^{r}}{r!}f_{r}(u)\\
=\sum_{i\geq 0}a_{i}(z-1)^{i}\cdot \sum_{j\geq 0}\frac{(z-1)^{j}}{j!}\sum_{m\geq 0}\frac{f_{j}^{(m)}(u)(z-1)^{m}u^{m}}{m!}-\sum_{r\geq 0}\frac{(z-1)^{r}}{r!}f_{r}(u)\\
=\sum_{s\geq 0}\left\{\sum_{i+j+m=s}\frac{a_{i}\cdot f_{j}^{(m)}\cdot u^{m}}{j!\cdot m!}-\frac{f_{s}(u)}{s!}\right\}(z-1)^{s},
\end{eqnarray*}
where to deduce the third step we used Taylor series, and then we replaced $z$ by $\sum_{i}a_{i}(z-1)^{i}$ where $a_0=a_1=1$ and $a_{i}=0$ for $i\geq 2$. The above after comparison of coefficients on both sides gives
$$f_{s-1}'(u)=(s-1)!\cdot \left\{\sum_{i+j+m=s}\frac{a_{i}\cdot f_{j}^{(m)}\cdot u^{m}}{j!\cdot m!}-\frac{f_{s}(u)}{s!}\right\},$$
for $s\geq 1$, which is
\begin{eqnarray*}
\label{eq4}
f_{s}'(u)=s!\cdot \left\{\sum_{i+j+m=s+1}\frac{a_{i}\cdot f_{j}^{(m)}\cdot u^{m}}{j!\cdot m!}-\frac{f_{s+1}(u)}{(s+1)!}\right\}\\
=s!\cdot \left\{\sum_{j+m=s+1}\frac{f_{j}^{(m)}\cdot u^{m}}{j!\cdot m!}+\sum_{j+m=s}\frac{f_{j}^{(m)}\cdot u^{m}}{j!\cdot m!}-\frac{f_{s+1}(u)}{(s+1)!}\right\}\\
=s!\left\{\sum_{m=0}^{s+1}\frac{f_{s+1-m}^{(m)}\cdot u^{m}}{m!\cdot (s+1-m)!}+\sum_{m=0}^{s}\frac{f_{s-m}^{(m)}\cdot u^{m}}{(s-m)!\cdot m!}-\frac{f_{s+1}(u)}{(s+1)!}\right\}\\
=s!\left\{\sum_{m=1}^{s+1}\frac{f_{s+1-m}^{(m)}\cdot u^{m}}{m!\cdot (s+1-m)!}+\sum_{m=0}^{s}\frac{f_{s-m}^{(m)}\cdot u^{m}}{(s-m)!\cdot m!}\right\}.
\end{eqnarray*}
for $s\geq 1$. This results in the following linear differential equation
\begin{equation}
\label{eq5}
f_{s}'(u)-\frac{f_{s}(u)}{(1-u)}=h_{s}(u)
\end{equation}
where 
$$
h_{s}(u)=\frac{1}{1-u}\left\{\frac{1}{s+1}\sum_{m=2}^{s+1}\binom{s+1}{m}f_{s+1-m}^{(m)}\cdot u^{m}+\sum_{m=1}^{s}\binom{s}{m}f_{s-m}^{(m)}\cdot u^{m}\right\}.
$$
The linear differential equation \eqref{eq5} can be solved as
$$f_{s}(u)=\frac{1}{1-u}\int_{0}^{u}h_{s}(t)(1-t)dt.$$
It can be proved using induction (see \cite{prodinger52} for similar calculation) that
$$f_{s}(u)=\frac{(2s)!}{4^{s}(1-u)^{2s+1}}-\frac{s}{9\cdot 4^{s-1}}(4s+5)\cdot \frac{(2s-1)!}{(1-u)^{2s}}+\mathfrak{R}_{0,2s-1}.$$
After this we conclude the result using equation \eqref{main}.
\end{proof}
\begin{remark}
See \cite{panny} for an alternate calculation.
\end{remark}
\end{theorem}
\section{Number of comparisons used by the randomized quick sort algorithm}
We assume that input arrays to the basic quick sort algorithm of size $n$ are permutations in the set $S_{n}$ with uniform probability distribution (each permutation occurring with probability $1/n!$\, .)\par 
Let $a_{n,k}$ be the number of permutations in $S_{n}$ requiring a total of $k$ comparisons to sort by the quick sort algorithm. Define $s$th factorial moment of the number of comparisons by
\begin{equation}
\beta_{s}(n)=\sum_{k\geq 0}(k)_{s}\frac{a_{n,k}}{n!},
\end{equation}
where $(k)_{s}=k(k-1)(k-2)\cdots (k-s+1).$\par 
It can be proved (see \cite{Sumit} for an alternate calculation for following) that
$$G_{n}(z)=\frac{z^{n-1}}{n}\sum_{1\leq j \leq n}G_{n-j}(z)G_{j-1}(z)$$
with $G_{0}(z)=1$. And from this 
$$\frac{\partial H(z,u)}{\partial u}= H^{2}(z,zu)$$
with $H(1,u)=(1-u)^{-1}$. Here $G_{n}(z)=\sum_{k\geq 0}\frac{a_{n,k}z^{k}}{n!}$ and $H(z,u)=\sum_{n\geq 0}G_{n}(z)u^{n}$ are corresponding probability generating function and bi-variate generating function, respectively.\par 
Let $f_{s}(u)=\sum_{n\geq 0}\beta_{s}(n)u^{n}$ be the generating functions for factorial moments. It can be proved using induction and 
$$f'_{s}(u)=s!\cdot \sum_{j+k+l+m=s}\frac{f^{(k)}_{j}(u) \cdot f_{l}^{(m)}(u)\cdot u^{k+m} }{j!\cdot k!\cdot  l!\cdot m!}$$
that
$$f_{s}(u)=\frac{2^{s}\, s!}{(1-u)^{s+1}}\left(\log\frac{1}{1-u}\right)^{s}+\frac{s(H_{s}-2)\, 2^{s}\, s!}{(1-u)^{s+1}}\left(\log\frac{1}{1-u}\right)^{s-1}+\mathfrak{R}_{s-2,s+1}(u),$$
after which using equation \ref{main} we can conclude
$$\beta_{s}(n)=2^{s}n^{s}\log^{s}{n}+2^{s}s(\gamma-2)n^{s}\log^{s-1}{n}+O(n^{s}\cdot \log^{s-2}{n}).$$
\bibliographystyle{amsplain}
\bibliography{sample.bib}
\end{document}